\documentclass[12pt]{iopart}
\usepackage{iopams,amsthm,epic,eepicemu}
\usepackage[pdfpagemode=UseOutlines,pdffitwindow=true,pdfstartview=Fit,
bookmarksnumbered=true,bookmarksopen=true,
pdfhighlight=/P,colorlinks=true,linkcolor=blue,urlcolor=blue,citecolor=blue]
{hyperref}

\theoremstyle{plain}
\newtheorem{theorem}{Theorem}
\newtheorem{statement}[theorem]{Statement}
\newtheorem{corollary}[theorem]{Corollary}
\theoremstyle{remark}
\newtheorem{remark}{Remark}
\theoremstyle{definition}
\newtheorem{definition}{Definition}

\newlength{\slen}
\newcommand{\sout}[1]{\settowidth{\slen}{$#1$}%
\hbox to 0pt{\rule[0.5ex]{\slen}{0.5pt}\hss}#1}

\def\ELB(#1,#2){\put(#1,#2){\color[gray]{0.5}\circle*{0.4}}%
{\thicklines\put(#1,#2){\circle{0.4}}}}
\def\ELW(#1,#2){\put(#1,#2){\color{white}\circle*{0.4}}%
{\thicklines\put(#1,#2){\circle{0.4}}}}

\newcommand{\expand}{\mathop{\rm expand}\nolimits}
\newcommand{\StirlingII}[2]%
{\Bigl\{\!\!{\footnotesize\begin{array}{c}#1\\#2\end{array}}\!\!\Bigr\}}

\begin{document}
\title{On the combinatorics of several integrable hierarchies}
\author{V E Adler}
\address{Landau Institute for Theoretical Physics,
Ak. Semenov str. 1-A, 142432 Chernogolovka, Russian Federation}
\ead{adler@itp.ac.ru}

\begin{abstract}\small
We demonstrate that statistics of certain classes of set partitions is
described by generating functions related to the Burgers, Ibragimov--Shabat and
Korteweg--de Vries integrable hierarchies.
\smallskip


\end{abstract}

\nosections


\hfill\begin{minipage}{110mm}\small
My procedure was this: I would count the stones by eye and write down the
figure. Then I would divide them into two handfuls that I would scatter
separately on the table. I would count the two totals, note them down, and
repeat the operation.\smallskip

\rightline{Borges, Blue tigers (translated by Andrew Hurley)}
\end{minipage}

\section{Introduction}

In the last decades numerous interrelations were discovered between the
combinatorics and the theory of integrable systems. Mainly, these links involve
{\em solutions}, either special ones, such as the Painlev\'e transcendents
\cite{Deift_2000} and solitons \cite{Grosset_Veselov_2005}, or generic ones,
such as the tau-function of the Kadomtsev--Petviashvili hierarchy
\cite{Alexandrov_Mironov_Morozov_Natanzon_2012}.

On the other hand, {\em equations} themselves exhibit a certain combinatorial
nature, due to the recurrent relations which govern the higher symmetries and
conservation laws of integrable hierarchies. This aspect was paid less
attention so far, although a quite simple description on the language of set
partitions was known for long for the Burgers hierarchy
\cite{Lambert_Loris_Springael_Willox_1994, Lambert_Springael_2008}. We
reproduce this combinatorial interpretation for the sake of completeness and as
a base for further generalizations. New results obtained in the paper are
related to the Ibragimov--Shabat and KdV hierarchies, see table
\ref{t:contents}. In these cases the combinatorics becomes more complicated,
since the ordinary set partitions are replaced by special ones which are
characterized by additional restrictions. Moreover, this combinatorics comes in
disguise: for instance, in the Burgers case we consider the generating function
intermediately for the higher flows, but in the KdV case we have to consider a
formal series for the logarithmic derivative of $\psi$-function which solves
the Riccati equation (inversion of the Miura map, see e.g.
\cite{Gelfand_Dikii_1975}). The flows and conservation laws of the hierarchy
are related with this generating function by simple algebraic relations. In the
Ibragimov--Shabat case a natural choice of the generating function is dictated
by the linearization procedure.

\begin{table}
\caption{\label{t:contents}Contents of the paper.}
\begin{indented}\small
\item[]\begin{tabular}{lp{85mm}}
\br
Hierarchy          & Combinatorial objects, their numbers\\
\mr
potential Burgers  & Set partitions, Bell polynomials $Y_n$,
                     Stirling numbers of the 2nd kind, Bell numbers\\[2pt]
Burgers            & Set partitions without distinguished singleton\\[2pt]
Ibragimov--Shabat  & $B$-type partitions, polynomials (\ref{an}),
                     $B$-analogs of Stirling numbers of the 2nd kind,
                     Dowling numbers\\[2pt]
Korteweg--de Vries & Non-overlapping partitions, polynomials (\ref{fn}),
                     number triangle (\ref{fu}), Bessel numbers\\
\br
\end{tabular}
\end{indented}
\end{table}

Although we are not interested in `explicit' formulae for the coefficients of
generating functions here, it should be mentioned that such formulae for the
potential KdV flows actually do exist. One of them, obtained already in
\cite{Gelfand_Dikii_1975}, represents the coefficient of a given monomial as a
certain multiple integral. Another formula obtained in \cite{Schimming_1995} is
of more combinatorial nature, but it remains very complicated. Only in the case
of pot-Burgers hierarchy the formula for the coefficients can be considered as
truly an explicit one.

\section{Potential Burgers hierarchy}

The pot-Burgers hierarchy is obtained from the linear heat equation hierarchy
\begin{equation}\label{heat}
 \psi_{t_n}=\psi_n
\end{equation}
by means of the change of dependent variable $\psi=e^v$. This yields
\begin{equation}\label{pot-B}
 v_{t_n}=e^{-v}D^n(e^v)=(D+v_1)^n(1)=Y_n(v_1,\dots,v_n),\quad n=0,1,2,\dots
\end{equation}
Here and further on we denote the derivatives as follows: $v_n=D^n(v)$,
$D=\partial/\partial x$, $v_{t_n}=\partial v/\partial t_n$. Several first
equations (\ref{pot-B}) are shown in the table \ref{t:pot-B}. A meaningful
combinatorics appears just from nothing!

\begin{table}
\caption{\label{t:pot-B} The potential Burgers hierarchy (weight $w(v_j)=j$).}
\begin{indented}
\item[]
\begin{eqnarray*}
\mr\\[-4pt]
 v_{t_0}= 1\\
 v_{t_1}= v_1\\
 v_{t_2}= v_2+v^2_1\\
 v_{t_3}= v_3+3v_1v_2+v^3_1\\
 v_{t_4}= v_4+(4v_1v_3+3v^2_2)+6v^2_1v_2+v^4_1\\
 v_{t_5}= v_5+(5v_1v_4+10v_2v_3)+(10v^2_1v_3+15v_1v^2_2)+10v^3_1v_2+v^5_1\\
\mr
\end{eqnarray*}
\end{indented}
\end{table}

It is easy to see that $Y_n$ are polynomials with integer coefficients,
homogeneous with respect to the weight $w(v_j)=j$. These polynomials play a
fundamental role in combinatorics and are known under the name of {\em (full
exponential) Bell polynomials} \cite{Comtet_1974}. An equivalent definition
through the exponential generating functions reads
\[
 \sum^\infty_{n=0}Y_n\frac{z^n}{n!}
  =e^{-v}\sum^\infty_{n=0}D^n(e^v)\frac{z^n}{n!}
  =e^{v(x+z)-v(x)}
  =\exp\left(\sum^\infty_{n=1}v_n\frac{z^n}{n!}\right)
\]
and this immediately implies the explicit formula
\begin{equation}\label{Yn}
 Y_n=\sum_{k_1+2k_2+\dots+rk_r=n}
 \frac{n!}{(1!)^{k_1}\dots(r!)^{k_r}k_1!\dots k_r!}\,v^{k_1}_1\dots v^{k_r}_r.
\end{equation}
Its combinatorial interpretation is obvious:

--- monomials correspond to partitions of the {\em number} $n$;

--- coefficients of monomials count partitions of the {\em set}
$[n]=\{1,\dots,n\}$ into the subsets (or blocks) of prescribed size.

For example, let us list all set partitions for $n=2,3,4$:

{\small
\begin{eqnarray*}
n=2:\quad
\begin{array}[t]{cc}
 v_2 & v^2_1 \\
  2  & 1+1 \\[3pt]
  12 & 1|2
\end{array}\qquad\qquad
n=3:\quad
\begin{array}[t]{ccc}
  v_3 & 3v_1v_2 & v^3_1 \\
  3   & 1+2     & 1+1+1 \\[3pt]
  123 & 1|23    & 1|2|3 \\
      & 2|13    &       \\
      & 3|12    &
\end{array}\\[10pt]
n=4:\quad
\begin{array}[t]{ccccc}
 v_4  & 4v_1v_3 & 3v^2_2 & 6v^2_1v_2 & v^4_1   \\
 4    & 1+3     & 2+2    & 1+1+2     & 1+1+1+1 \\[3pt]
 1234 & 1|234   & 12|34  & 1|2|34    & 1|2|3|4 \\
      & 2|134   & 13|24  & 1|3|24    &         \\
      & 3|124   & 14|23  & 1|4|23    &         \\
      & 4|123   &        & 2|3|14    &         \\
      &         &        & 2|4|13    &         \\
      &         &        & 3|4|12    &
\end{array}
\end{eqnarray*}}

Recall, that each set partition is considered as unordered set (with blocks as
the elements), that is, ordering of the blocks does not matter. However, it is
often useful to enumerate the blocks somehow. For the sake of definiteness, we
will adopt the enumeration corresponding to the ordering of the minimal
elements in the blocks.

We see that the combinatorics behind the hierarchy (\ref{pot-B}) is quite
simple. The following statement is well known, see e.g.
\cite{Lambert_Loris_Springael_Willox_1994, Lambert_Springael_2008}.

\begin{theorem}\label{th:pot-B}
In the potential Burgers hierarchy, the coefficient of the monomial
$v^{k_1}_1\dots v^{k_r}_r$ is equal to the number of partitions of the set of
$n=k_1+2k_2+\dots+rk_r$ elements into $k_1$ blocks with 1 element, $k_2$ blocks
with 2 elements, \dots, $k_r$ blocks with $r$ elements.
\end{theorem}
\begin{proof}
One proof follows intermediately from the explicit formula (\ref{Yn}) for the
coefficients. However, we will not always have such a formula at hand. The
following reasoning provides a more conceptual proof.

Let $\Pi_{n,k}$ denotes the set of all partitions of the set $[n]$ into $k$
blocks and $\Pi_n$ denotes the set of all partitions of $[n]$. Let us consider
operations
\[
 d_j:\Pi_{n,k}\to \Pi_{n+1,k},~~ j=1,\dots,k,\qquad
 M:\Pi_{n,k}\to \Pi_{n+1,k+1}
\]
defined, respectively, as appending of the element $n+1$ to $j$-th block or
adding it to the partition as a new singleton. This can be visualized by the
following diagram:
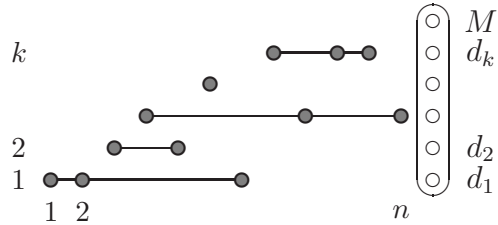
\begin{center}
\begin{picture}(15,7.5)(-0.5,-0.5)
 \path(1,1)(7,1) \ELB(1,1)\ELB(2,1)\ELB(7,1)
 \path(3,2)(5,2) \ELB(3,2)\ELB(5,2)
 \path(4,3)(12,3)\ELB(4,3)\ELB(9,3)\ELB(12,3)
                 \ELB(6,4)
 \path(8,5)(11,5)\ELB(8,5)\ELB(10,5)\ELB(11,5)
 \put(13,3.5){\oval(1,6)}
 \multiput(13,1)(0,1){6}{\circle{0.4}}
{\small
\multiputlist(1,0)(1,0){$1$,$2$,,,,,,,,,,$n$}
\multiputlist(0,1)(0,1){$1$,$2$,,,$k$}}
\multiputlist(14.5,1)(0,1){$d_1$,$d_2$,,,$d_k$,$M$}
\end{picture}
\end{center}
Starting from the partition $\{\varnothing\}$ of the set $[0]=\varnothing$ and
applying operations $d_j,M$, one can generate, in a unique way, any partition
of $[n]$. Indeed, the required sequence of operations is uniquely recovered by
deleting elements in the inverse order from $n$ to $1$.

In the theorem, a set partition $\pi$ with $k_1$ 1-blocks, \dots, $k_r$
$r$-blocks corresponds to the monomial $p(\pi)=v^{k_1}_1\dots v^{k_r}_r$. The
differentiation $D(p(\pi))$ by the Leibnitz rule amounts to replacing of $v_i$
with $v_{i+1}$ for each factor in turn, taking the multiplicity into account.
In the partition language, this means that we add the new element to each block
in turn. As the result, we obtain the sum of monomials $p(d_j(\pi))$ for all
admissible values of $j$. Multiplication of the monomial $p(\pi)$ by $v_1$
gives the monomial $p(M(\pi))$. Thus, the polynomials
\[
 P_n=\sum_{\pi\in\Pi_n}p(\pi)
\]
are related by the recurrent relation $P_{n+1}=(D+v_1)(P_n)$ and since
$P_1=v_1$, hence $P_n=Y_n(v_1,\dots,v_n)$.
\end{proof}

A less detailed statistics is obtained if we forget about sizes of blocks and
consider just their number in a given partition. Obviously, this correspond to
summing up the coefficients of terms of the same degree, which gives us the
{\em Bell polynomials} of one variable
\[
 B_n(u)=Y_n(u,\dots,u)=(u\partial_u+u)^n(1)=\sum^n_{k=0}\StirlingII{n}{k}u^k.
\]
The coefficient $\StirlingII{n}{k}$ of $u^k$, that is, the number of partitions
of $[n]$ into $k$ blocks, is called the {\em Stirling number of the second
kind} \cite[A048993]{OEIS}:

{\small
\[
 \begin{array}{llllllllll}
  1&  &   &    &    &    &   &   && 1\\
  0& 1&   &    &    &    &   &   && 1\\
  0& 1& 1 &    &    &    &   &   && 2\\
  0& 1& 3 & 1  &    &    &   &   && 5\\
  0& 1& 7 & 6  & 1  &    &   &   && 15\\
  0& 1& 15& 25 & 10 & 1  &   &   && 52\\
  0& 1& 31& 90 & 65 & 15 & 1 &   && 203\\
  0& 1& 63& 301& 350& 140& 21& 1 && 877
 \end{array}
\]}

\noindent By definition, $\StirlingII{n}{0}=0$ at $n>0$ and
$\StirlingII{0}{0}=\#\{\varnothing\}=1$. The total numbers of set partitions
with $n$ elements, the Bell or the exponential numbers \cite[A000110]{OEIS},
are given by the sums of the rows:
\[
 B_n=B_n(1)=Y_n(1,\dots,1)=\sum^n_{k=0}\StirlingII{n}{k},\quad
 \sum^\infty_{n=0}B_n\frac{z^n}{n!}=e^{e^z-1}.
\]

\section{Burgers hierarchy}

The right hand sides of equations (\ref{pot-B}) do not contain $v$ and this
makes the substitution $u=v_1$ possible. This brings to the Burgers hierarchy
\begin{equation}\label{B}
 u_{t_n}=D(Y_n(u,\dots,u_{n-1})),\quad n=1,2,\dots
\end{equation}
which is homogeneous with respect to the weight $w(u_j)=j+1$. Several first
equations are written down in the table \ref{t:B}. What is the combinatorial
interpretation in this case? This can be easily understood by the following
example, for $n=3$:

{\small
\[
\begin{array}{ccccccccc}
 v_3&3v_1v_2&v^3_1& &    &            &      &             & \\
 u_2& 3uu_1 & u^3 &\quad\stackrel{D}{\longrightarrow}\quad
                    & u_3&3uu_2       &3u^2_1& 3u^2u_1     &0u^4          \\[3pt]
 123& 1|23  &1|2|3& &1234&      1|234 &12|34 &      1|2|34 &\sout{1|2|3|4}\\
    & 2|13  &     & &    &      2|134 &13|24 &      1|3|24 &              \\
    & 3|12  &     & &    &      3|124 &14|23 &\sout{1|4|23}&              \\
    &       &     & &    &\sout{4|123}&      &      2|3|14 &              \\
    &       &     & &    &            &      &\sout{2|4|13}&              \\
    &       &     & &    &            &      &\sout{3|4|12}&
\end{array}
\]}

\noindent Certainly, renaming $v_j\to u_{j-1}$ does not change the combinatorics. The
differentiation amounts to appending the new element to all blocks in turn,
however, now we do not add it as a new block. Therefore, the partitions under
consideration are constructed as in theorem \ref{th:pot-B}, but we do not apply
the operation $M$ at the last step. As a result, all partitions $\Pi_n$ are
mapped onto some subset of partitions $\Pi_{n+1}$, namely, those partitions
where the element $n+1$ does not appear as a singleton. We arrive to the
following combinatorial interpretation of equations (\ref{B}).

\begin{table}
\caption{\label{t:B} Burgers hierarchy (weight $w(u_j)=j+1$).}
\begin{indented}
\item[]
\begin{eqnarray*}
\mr\\[-4pt]
 u_{t_1}= u_1\\
 u_{t_2}= u_2+2uu_1\\
 u_{t_3}= u_3+(3uu_2+3u^2_1)+3u^2u_1\\
 u_{t_4}= u_4+(4uu_3+10u_1u_2)+(6u^2u_2+12uu^2_1)+4u^3u_1\\
 u_{t_5}= u_5+(5uu_4+15u_1u_3+10u^2_2)+(10u^2u_3+50uu_1u_2+15u^3_1)\\
  \qquad\qquad +(10u^3u_2+30u^2u^2_1)+5u^4u_1\\
\mr
\end{eqnarray*}
\end{indented}
\end{table}

\begin{theorem}
In the Burgers hierarchy, the coefficient of the monomial
$u^{k_0}u^{k_1}_1\dots u^{k_r}_r$ is equal to the number of partitions of the
set with one distinguished element into $k_0$ blocks with 1 element, \dots,
$k_r$ blocks with $(r+1)$ element and such that the distinguished element does
not constitute 1-block.
\end{theorem}

As before, one can consider more rough statistics. For instance, setting $u=1$
gives us the total number of partitions under consideration of the set $[n+1]$:
\[
 D(Y_n(u,\dots,u_{n-1}))|_{u_j=1}=B'_n(1)=\sum^n_{k=1}k\StirlingII{n}{k},
 \quad n\ge1.
\]
The sequence of these numbers ({\em 2-Bell numbers}) starts

{\small
\[
 1,~ 3,~ 10,~ 37,~ 151,~ 674,~ 3263,~ 17007,~ 94828,~ 562595,~ \dots
\]}

\noindent
According to \cite[A005493]{OEIS}, it can be characterized also in many other
ways, in particular, as the number of partitions of $[n]$ with distinguished
block or as the total number of blocks in all set partitions of $[n]$. These
interpretations are obvious as well, since the distinguished blocks can be
identified with the blocks enlarged by the operations $d_j$, and these
operations are applied exactly as many times as there are blocks in all
partitions.

\section{Ibragimov--Shabat hierarchy}

\subsection{Recurrent relations}

The table \ref{t:IS} displays the sequence of point changes and substitutions
between equation $\psi_{t_3}=\psi_3$ and the Ibragimov--Shabat equation
\cite{Ibragimov_Shabat_1980}
\[
 u_{t_3}=u_3+3u^2u_2+9uu^2_1+3u^4u_1.
\]
Although this transformation looks quite harmless, it partially destroys the
symmetry algebra: in the variables $\psi$, it consists of equations
(\ref{heat}) of arbitrary order, while only odd order equations survive in the
variables $u$. Indeed, the change $\psi^2=s$ brings to equation $s_{t_n}=\dots$
where the right hand side is a full derivative only if $n$ is odd:
\begin{equation}\label{stn}
 s_{t_n}=2\psi\psi_n
  =D\bigl(2\psi\psi_{n-1}-2\psi_1\psi_{n-2}+2\psi_2\psi_{n-3}
   +\dots\pm\psi^2_{(n-1)/2}\bigr).
\end{equation}
In the analogous equation for even $n$ the term $\psi^2_{n/2}$ remains outside
the parentheses, that is, $s_{t_n}\not\in\mathop{\rm Im}D$, and therefore the
further substitution $s=q_1$ leads out of the class of evolutionary equations.
The structure of odd flows is described by the following statement.

\begin{table}
\caption{\label{t:IS} Linearization of the Ibragimov--Shabat equation.}
\normalsize
\[
\begin{array}{ccc}
 \mr
 \psi_{t_3}=\psi_3 & &
 u_{t_3}=u_3+3u^2u_2+9uu^2_1+3u^4u_1\\[5pt]
 \updownarrow \psi^2=s && \updownarrow u^2=v \\[5pt]
 s_{t_3}=D\left(s_2-\frac{3s^2_1}{4s}\right) &&
 v_{t_3}=D\left(v_2-\frac{3v^2_1}{4v}+3vv_1+v^3\right)\\[5pt]
 \uparrow s=q_1 && \uparrow v=w_1 \\[5pt]
 q_{t_3}=q_3-\frac{3q^2_2}{4q_1} &
  \stackrel{q=e^{2w}}{\leftarrow\!\!\!-\!\!-\!\!\!\rightarrow} &
 w_{t_3}=w_3-\frac{3w^2_2}{4w_1}+3w_1w_2+w^3_1\\
\mr
\end{array}
\]
\end{table}

\begin{statement}
Let us denote $D_t=\partial_{t_1}+z^2\partial_{t_3}+z^4\partial_{t_5}+\dots$,
$A=A(z)=a_0+a_1z+a_2z^2+\dots$, $\bar A=A(-z)$, then the Ibragimov--Shabat
hierarchy is equivalent to equations
\begin{eqnarray}
\label{ut}
 D_t(u)=\frac{1}{2u}D(A\bar A)=\frac{1}{2z}(A-\bar A)-uA\bar A,\\
\label{A}
 z(D+u^2)(A)=A-u.
\end{eqnarray}
\end{statement}

\begin{proof}
Let us consider the generating function
\[
 \Psi=\psi+\psi_1z+\psi_2z^2+\dots
\]
and set $\Psi=\sqrt{2}e^wA$. Equation (\ref{A}) follows from the relations
\[
 zD(\Psi)=\Psi-\psi,\qquad \psi=\sqrt{q_1}=\sqrt{2e^{2w}w_1}=\sqrt{2}e^wu.
\]
Next, let $\bar\Psi=\Psi(-z)$, then (cf (\ref{stn}))
\begin{eqnarray*}
 D(\Psi\bar\Psi)&=z^{-1}(\Psi-\psi)\bar\Psi-z^{-1}\Psi(\bar\Psi-\psi)\\
  &=z^{-1}\psi(\Psi-\bar\Psi)
  =2\psi(\psi_1+\psi_3z^2+\dots)=2\psi D_t(\psi)=D_t(s).
\end{eqnarray*}
Applying $D^{-1}$ yields $\Psi\bar\Psi=D_t(q)=2e^{2w}D_t(w)$, wherefrom
\[
 2uD_t(u)=D_t(v)=DD_t(w)=\frac{1}{2}D(e^{-2w}\Psi\bar\Psi)=D(A\bar A).
\]
Second equality in (\ref{ut}) follows after elimination of derivatives by use
of (\ref{A}).
\end{proof}

Equation (\ref{A}) is equivalent to recurrent relations
\begin{equation}\label{an}
 a_0=u,\quad a_n=a_n(u,\dots,u_n)=(D+u^2)(a_{n-1}),\quad n=1,2,\dots
\end{equation}
which are our object of study. Let us try to find a combinatorial
interpretation for this recursion.

\begin{table}
\caption{\label{t:an} Polynomials $a_n$ (weight $w(u_j)=2j+1$).}
\begin{indented}
\item[]
\begin{eqnarray*}
\mr\\[-4pt]
\fl a_0= u \\
\fl a_1= u_1+u^3\\
\fl a_2= u_2+4u^2u_1+u^5\\
\fl a_3= u_3+(5u^2u_2+8uu^2_1)+9u^4u_1+u^7\\
\fl a_4= u_4+(6u^2u_3+26uu_1u_2+8u^3_1)+(14u^4u_2+44u^3u^2_1)+16u^6u_1+u^9 \\
\fl a_5= u_5+(7u^2u_4+38uu_1u_3+26uu^2_2+50u^2_1u_2)
 +(20u^4u_3+170u^3u_1u_2+140u^2u^3_1)\\
 +(30u^6u_2+140u^5u^2_1)+25u^8u_1+u^{11}\\
\mr
\end{eqnarray*}
\end{indented}
\end{table}

Several first polynomials $a_n$ are presented in the table \ref{t:an}. Given
these data as a prescribed statistics, our goal is to figure out a definition
of the corresponding combinatorial objects, that is, to solve a kind of inverse
problem of the enumerative combinatorics. In contrast to the Burgers hierarchy
case, here we do not know an explicit formula for the coefficients, but this is
not too important, the main problem is to guess what are the objects which we
are counting. An invaluable aid in such an ill-posed problem may be obtained by
comparison with the known data collected in the Encyclopedia of integer
sequences \cite{OEIS}. Let us pass to the less detailed statistics by gluing
together terms of the same degree. Polynomials of one variable
$a_n(u,\dots,u)=(u\partial_u+u^2)^n(u)$ contain only odd powers of $u$ and
their coefficients constitute the triangle

{\small
\[
 \begin{array}{llllllllll}
 1 &     &     &      &     &    &   &  && 1\\
 1 & 1   &     &      &     &    &   &  && 2\\
 1 & 4   & 1   &      &     &    &   &  && 6\\
 1 & 13  & 9   & 1    &     &    &   &  && 24\\
 1 & 40  & 58  & 16   & 1   &    &   &  && 116\\
 1 & 121 & 330 & 170  & 25  & 1  &   &  && 648\\
 1 & 364 & 1771& 1520 & 395 & 36 & 1 &  && 4088\\
 1 & 1093& 9219& 12411& 5075& 791& 49& 1&& 28640
 \end{array}
\]}

\noindent
which turns out to be known: according to \cite[A039755]{OEIS} this is the
triangle of analogs of Stirling numbers of the second kind for the so-called
$B$ type set partitions. The sums of numbers in rows, that is, the total sums
of the coefficients $a_n(1,\dots,1)$ form the sequence \cite[A007405]{OEIS} of
the {\em Dowling numbers}, or {\em $B$-analogs of the Bell numbers}. This gives
us a broad hint at a possible connection between polynomials (\ref{an}) and $B$
type partitions. This guess is proved in the next section.

\subsection{Generating operations for type $B$ set partitions}

Special classes of set partitions appear when one takes into account some
additional structure of the set. Set partitions of $B$ type (or signed set
partitions, $\mathbb Z_2$-partitions) \cite{Dowling_1973}, see also e.g.
\cite{Benoumhani_1996, Suter_2000} make use of the reflection $j\to-j$.

\begin{definition}
A partition $\pi$ of the set $\{-n,\dots,n\}$ is called the $B_n$ type
partition if:

1) $\pi=-\pi$, that is for each block $\beta\in\pi$ also $-\beta\in\pi$;

2) $\pi$ contains only one block $\pi_0\in\pi$ such that $\pi_0=-\pi_0$.

We will denote $\Pi^B_n$ the set of all such partitions and $\Pi^B_{n,k}$ those
partitions which contain $k$ block pairs.
\end{definition}

In a brief notation for $B$ type partitions, the negative elements of the
$0$-block are omitted, and only that block of each pair is displayed for which
the element with minimal absolute value is positive; the minus signs are
denoted by over bars. For instance, in this notation the partition
$-5,-4|{-3,0,3}|{-2,1}|{-1,2}|4,5$ is represented as $03|1\bar 2|45$. A
graphical representation is clear from the diagram
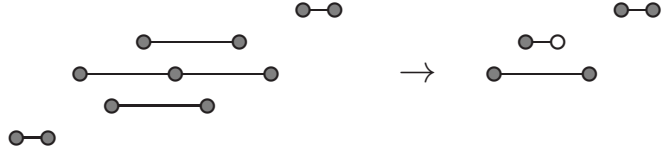
\begin{center}
\begin{picture}(20,5.5)(-5,-2.5)
 \path(-5,-2)(-4,-2)\ELB(-5,-2)\ELB(-4,-2)
 \path(5,2)(4,2)    \ELB(5,2)\ELB(4,2)
 \path(-2,-1)(1,-1) \ELB(-2,-1)\ELB(1,-1)
 \path(2,1)(-1,1)   \ELB(2,1)\ELB(-1,1)
 \path(-3,0)(3,0)\ELB(-3,0)\ELB(0,0)\ELB(3,0)
 \put(7,-0.25){\large $\to$}
 \put(10,0){
 \path(5,2)(4,2)\ELB(5,2)\ELB(4,2)
 \path(1,1)(2,1)\ELB(1,1)\ELW(2,1)
 \path(0,0)(3,0)\ELB(0,0)\ELB(3,0)}
\end{picture}
\end{center}

Now let us define the map $p$ from $\Pi^B_n$ into the set of monomials on the
variables $u_j$. Let $|\beta|$ denote the number of positive elements in the
block $\beta$:
\[
 |\beta|=\#\{i\in\beta:i>0\}.
\]
It is clear that the number of negative elements in the block is $|\bar\beta|$.
Let a set partition $\pi\in\Pi^B_{n,k}$ consists of 0-block $\pi_0$ and block
pairs $\pi_1,\bar\pi_1$, \dots, $\pi_k,\bar\pi_k$, such that the element of
$\pi_j$ with minimal absolute value is positive. For such a partition, let
\[
 p(\pi)=u_{|\pi_0|}\cdot u_{|\pi_1|-1}u_{|\bar\pi_1|}\cdots
  u_{|\pi_k|-1}u_{|\bar\pi_k|}.
\]
As an example, let us write down $\Pi^B_3$ partitions, collecting together all
partitions corresponding to the same monomial:

{\small
\[
\begin{array}[t]{ccccc}
 u_3  & 5u^2u_2       & 8uu^2_1   & 9u^4u_1    & u^7     \\[5pt]
 0123 & 0|123         & 0|12\bar3 & 0|12|3     & 1|2|3|4 \\
      & 0|1\bar2\bar3 & 0|1\bar23 & 0|1\bar2|3 &         \\
      & 012|3         & 01|23     & 0|13|2     &         \\
      & 013|2         & 01|2\bar3 & 0|1\bar3|2 &         \\
      & 023|1         & 02|13     & 0|23|1     &         \\
      &               & 02|1\bar3 & 0|2\bar3|1 &         \\
      &               & 03|12     & 01|2|3     &         \\
      &               & 03|1\bar2 & 02|1|3     &         \\
      &               &           & 03|1|2     &
\end{array}
\]}

\noindent
The resulting polynomial is exactly $a_3$. The following theorem demonstrates
that this is not just a coincidence and the polynomials $a_n$ are, indeed, the
$\mathbb Z_2$-analogs of the full exponential Bell polynomials $Y_n$.

\begin{theorem}
The polynomials (\ref{an}) are equal to
\[
 a_n=\sum_{\pi\in\Pi^B_n}p(\pi).
\]
\end{theorem}
\begin{proof}
Let us denote the sum in the right hand side $p_n$. Obviously, $p_0=u=a_0$, so
we only have to prove that $p_n$ satisfy the same recurrent relations as $a_n$,
that is, $p_n=(D+u^2)(p_{n-1})$.

Notice, that deleting of elements $\pm n$ from any $B_n$ type set partition
gives us a $B_{n-1}$ type set partition. Therefore, $\Pi^B_n$ is constructed
from $\Pi^B_{n-1}$ by adding $\pm n$ in all possible ways. It is easy to see
that this is done by the following operations:
\begin{list}{}{}
\item $d_0:\Pi^B_{n-1,k}\to\Pi^B_{n,k}$,
      insertion of both elements $\pm n$ into 0-block;
\item $d_j:\Pi^B_{n-1,k}\to\Pi^B_{n,k}$, $j=1,\dots,k$,
      insertion of $\pm n$ into blocks $\pm\pi_j$;
\item $\bar d_j:\Pi^B_{n-1,k}\to\Pi^B_{n,k}$, $j=1,\dots,k$,
      insertion of $\pm n$ into blocks $\mp\pi_j$;
\item $M:\Pi^B_{n-1,k}\to\Pi^B_{n,k+1}$,
      adding of the new block pair $\{-n\}$, $\{n\}$.
\end{list}
Starting from the trivial partition of the set $\{0\}$ and applying these
operations, one can obtain, in a unique way, any $B$ type set partition. Let us
keep track of the monomial $p(\pi)$, $\pi\in\Pi^B_{n-1,k}$ under these
operations:
\begin{list}{}{}
\item $d_0:$ the factor $u_{|\pi_0|}$ is replaced with $u_{|\pi_0|+1}$;
\item $d_j:$ the factor $u_{|\pi_j|-1}$ is replaced with $u_{|\pi_j|}$;
\item $\bar d_j:$ the factor $u_{|\bar\pi_j|}$
      is replaced with $u_{|\bar\pi_j|+1}$;
\item $M:$ two new factors $u$ are added.
\end{list}
Therefore, application of all possible operations maps the monomial $p(\pi)$ to
the sum of monomials $(D+u^2)(p(\pi))$.
\end{proof}

\section{Korteweg--de Vries hierarchy}

\subsection{Recurrent relations}

Let us recall (for a proof, see e.g. \cite{Gelfand_Dikii_1975}) a computation
method of the KdV conservation laws and flows, based on solving of the Riccati
equation
\begin{equation}\label{Df}
 D(f)+f^2=\lambda-u,\quad \lambda=z^2/4
\end{equation}
by the formal power series
\[
 f(z)=-\frac{z}{2}+\frac{f_1(u)}{z}+\frac{f_2(u,u_1)}{z^2}
  +\dots +\frac{f_n(u,\dots,u_{n-1})}{z^n}+\cdots
\]
Equation (\ref{Df}) is equivalent to the recurrent relations
\begin{equation}\label{fn}
 f_1=u,\quad
 f_{n+1}=D(f_n)+\sum^{n-1}_{s=1}f_sf_{n-s},\quad n=1,2,\dots
\end{equation}
which will be the main object of our study. Several polynomials $f_n$ are
written down in the table \ref{t:KdV}. The flows are computed from the
polynomials with odd subscripts: let
\[
 g(z)= -\frac{1}{2z}-\frac{g_1}{z^3}-\frac{g_3}{z^5}-\dots
   -\frac{g_{2m-1}}{z^{2m+1}}-\cdots= \frac{1}{2(f(z)-f(-z))}
\]
which is equivalent to recurrent relations
\[
 g_1=u,\quad
 g_{2m+1}=f_{2m+1}+2\sum^m_{s=1}g_{2s-1}f_{2m-2s+1},\quad m=1,2,\dots
\]
then the KdV hierarchy reads
\[
 u_{t_{2m+1}}=D(g_{2m+1})=u_{2m+1}+\dots,\quad m=0,1,2,\dots
\]
Moreover, polynomials (\ref{fn}) with odd subscripts serve as common conserved
densities for all these flows.

\begin{table}
\caption{\label{t:KdV}Polynomials $f_n$ (weight $w(u_j)=j+2$).}
\begin{indented}\item[]
\begin{eqnarray*}
\mr\\[-4pt]
 f_1= u\\
 f_2= u_1\\
 f_3= u_2+u^2\\
 f_4= u_3+4uu_1\\
 f_5= u_4+(6uu_2+5u^2_1)+2u^3\\
 f_6= u_5+(8uu_3+18u_1u_2)+16u^2u_1\\
 f_7= u_6+(10uu_4+28u_1u_3+19u^2_2)+(30u^2u_2+50uu^2_1)+5u^4\\
\mr
\end{eqnarray*}
\end{indented}
\end{table}

One interpretation of the polynomials $f_n$ can be seen intermediately from the
recurrent relations (\ref{fn}). Let us consider expressions $\varphi$ builded
from the variable $u$ and operations $M(a,b)$, $d_j(a)$, $1\le j\le\deg a$
where $\deg a$ is equal to the number of instances of $u$ in $a$. Such
expressions can be called `unexpanded monomials'. For any expression $\varphi$
its value $\expand(\varphi)$ is computed according to the following rules:
\begin{list}{}{}
\item[---] independently on the order of operations, all $d_j$ are applied before $M$;
\item[---] the action of $d_j(a)$ amounts to replacing of $j$-th instance of $u_i$
      in $a$ with $u_{i+1}$ ($u$ is identified with $u_0$, as usual);
\item[---] $M(a,b)$ is replaced by the product $ab$.
\end{list}
Let $\Phi_n$ denote the set of all expressions with the total number of symbols
$u,d,M$ equal to $n$. For instance:

{\small
\[
\begin{array}{lllll}
\br
     && \mbox{unexpanded monomials} && \mbox{expanded monomials}\\
\mr
 n=1 && u                           && u          \\
 n=2 && d_1(u)                      && u_1        \\
 n=3 && d_1(d_1(u)),~~ M(u,u )      && u_2,~~u^2  \\
 n=4 && d_1(d_1(d_1(u))),           && u_3,       \\
     && d_1(M(u,u)),~~ d_2(M(u,u))  && uu_1,~uu_1 \\
     && M(d_1(u),u),~~ M(u,d_1(u))  && uu_1,~uu_1 \\
\br
\end{array}
\]}

\begin{theorem}\label{th:KdV1}
The number of different expressions builded from symbols $M,d_j,u$ with the
same monomial as their value is equal to the coefficient of this monomial in
polynomials $f_n$. In other words,
\begin{equation}\label{fphi}
 f_n=\sum_{\varphi\in\Phi_n}\expand(\varphi).
\end{equation}
\end{theorem}

\begin{proof}
Any expression from $\Phi_{n+1}$, $n>0$ is either of the form $d_j(a)$ where
$a\in\Phi_n$, $1\le j\le\deg a$ or of the from $M(a,b)$ where $a\in\Phi_s$,
$b\in\Phi_{n-s}$. Taking into account the obvious properties
\begin{eqnarray*}
 \sum^{\deg a}_{j=1}\expand(d_j(a))=D(\expand(a)),\\
 \expand(M(a,b))=\expand(a)\expand(b),
\end{eqnarray*}
this implies that polynomials (\ref{fphi}) satisfy the recurrent relation
(\ref{fn}).
\end{proof}

This interpretation is fairly intuitive, but it is desirable to compare it with
something more standard. As before, let us pass to polynomials of one variable
by collecting together terms of the same degrees. This brings us to a number
triangle which apparently is not in the OEIS:

\begin{equation}\label{fu}
\small
\begin{array}{llllllll}
 1 &      &       &       &      &    && 1 \\
 1 &      &       &       &      &    && 1 \\
 1 & 1    &       &       &      &    && 2 \\
 1 & 4    &       &       &      &    && 5 \\
 1 & 11   & 2     &       &      &    && 14 \\
 1 & 26   & 16    &       &      &    && 43 \\
 1 & 57   & 80    & 5     &      &    && 143 \\
 1 & 120  & 324   & 64    &      &    && 509 \\
 1 & 247  & 1170  & 490   & 14   &    && 1922 \\
 1 & 502  & 3948  & 2944  & 256  &    && 7651 \\
 1 & 1013 & 12776 & 15403 & 2730 & 42 && 31965
\end{array}
\end{equation}

\noindent
Nevertheless, the sequence of coefficients sum totals turns out to be known:
$f_{n+1}[1]$ is equal to the number of {\em non-overlapping partitions} of the
set $[n]$, or the {\em Bessel number} $B^*_n$ \cite[A006789]{OEIS}. Notice,
that identifying of all $u_j$ results in the Riccati equation
$u\partial_u(f)+f^2=\lambda-u$ which is, indeed, equivalent to the Bessel
equation. Moreover, one can see in the triangle the Euler numbers
\cite[A000295]{OEIS}, the Catalan numbers \cite[A000108]{OEIS} and powers of 4.

\subsection{Generating operations for non-overlapping partitions}

This class of set partitions was introduced in \cite{Flajolet_Schott_1990}, see
also \cite{Claesson_2001,Klazar_2003}. Its definition engages the order
relation on the partitioned set $[n]=\{1,\dots,n\}$.

\begin{definition}
Blocks $\alpha$ and $\beta$ of a set partition $\pi$ overlap if
\[
 \min\alpha<\min\beta<\max\alpha<\max\beta.
\]
The set partition is called non-overlapping (NOP) if any two blocks in it do
not overlap. All NOPs of the set $[n]$ will be denoted $\Pi^*_n$.
\end{definition}

The interval $[\min\alpha,\max\alpha]$ is called the {\em support} of the block
$\alpha$. The above definition of NOP is equivalent to the property that
supports of any two blocks either do not intersect or lie one in another. The
left diagram below shows overlapping blocks and the right diagram shows
non-overlapping ones:
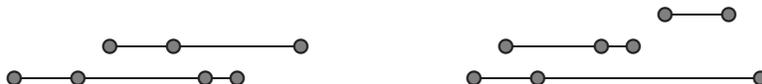
\begin{center}
\begin{minipage}[t]{10pc}
\begin{picture}(9,3)(0,-0.7)
 \path(0,0)(7,0)  \ELB(0,0)\ELB(2,0)\ELB(6,0)\ELB(7,0)
 \path(3,1)(9,1) \ELB(3,1)\ELB(5,1)\ELB(9,1)
\end{picture}
\end{minipage}
\qquad\qquad
\begin{minipage}[t]{10pc}
\begin{picture}(9,3)(0,-0.7)
 \path(0,0)(9,0) \ELB(0,0)\ELB(2,0)\ELB(9,0)
 \path(1,1)(5,1) \ELB(1,1)\ELB(4,1)\ELB(5,1)
 \path(6,2)(8,2) \ELB(6,2)\ELB(8,2)
\end{picture}
\end{minipage}
\end{center}

\begin{remark}
A neighbour class of {\em non-crossing} partitions is characterized by a more
restrictive condition which forbids the pattern
$\alpha_1<\beta_1<\alpha_2<\beta_2$ for any elements of any two blocks. It is
under active study in combinatorics as well, moreover, it makes sense to
combine such types of restrictions with symmetries like the reflection for the
$B$ type partitions, see e.g. \cite{Reiner_1997}. It is an open question,
whether some integrable hierarchies may be associated with such kind of
objects.
\end{remark}

Some simple properties of NOPs are the following.

--- At $n=0,1,2,3$ we have $\Pi^*_n=\Pi_n$ and there is only one overlapping
partition $13|24$ in $\Pi_4$,.

--- Singletons do not overlap with any block.

--- NOPs containing only doublets can be easily identified with the balanced
sets of parentheses:
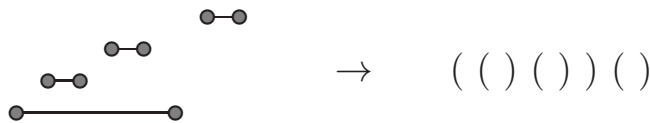
\begin{center}
\begin{picture}(20,4)(0,-0.5)
 \path(0,0)(5,0)\ELB(0,0)\ELB(5,0)
 \path(1,1)(2,1)\ELB(1,1)\ELB(2,1)
 \path(3,2)(4,2)\ELB(3,2)\ELB(4,2)
 \path(6,3)(7,3)\ELB(6,3)\ELB(7,3)
 \put(10,1){\large $\to\qquad (~(~)~(~)~)~(~)$}
\end{picture}
\end{center}

The last property explains where from the Catalan numbers appear in the
triangle (\ref{fu}). The deletion of differentiation in equation (\ref{fn})
brings to the recursion for the `dispersionless terms':
\[
 f_1=u,~~ f_{n+1}=\sum^{n-1}_{s=1}f_sf_{n-s}
 ~~\to~~ u,0,u^2,0,2u^3,0,5u^4,0,~\dots
\]

In order to establish a correspondence with the general polynomials $f_n$, let
us identify the variable $u$ with the set partition $\{\varnothing\}$ and
define the action of the operations $M$ and $d_j$ on the NOPs, in such a way
that expressions $\Phi_{n+1}$ be in a one-to-one correspondence with $\Pi^*_n$.
\smallskip

{\em Degree.} Let $\deg\pi=k$ if $\pi$ contains $k-1$ multiplets (blocks with
more than one element).
\smallskip

{\em Operation $M$.} Let $\rho\in\Pi^*_r$, $\sigma\in\Pi^*_s$. Denote by
$(\sigma)_{r+1}$ the partition of the set $\{r+2,r+s+1\}$ obtained from
$\sigma$ by adding $r+1$ to each element, and define
\[
 M(\rho,\sigma)=\rho\cup\{\{r+1,r+s+2\}\}\cup(\sigma)_{r+1}\in\Pi^*_{r+s+2}.
\]
This can be illustrated by the diagram
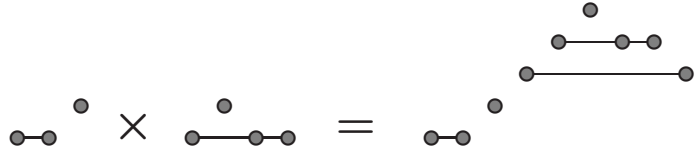
\begin{center}
\begin{picture}(20,5.5)(0,-1)
\path(0,0)(1,0)\ELB(0,0)\ELB(1,0)\ELB(2,1)%
\put(3,0){\Large$\boldsymbol\times$}
\put(5.5,0){\path(0,0)(3,0)\ELB(0,0)\ELB(2,0)\ELB(3,0)\ELB(1,1)}%
\put(10,0){\Large$\boldsymbol=$}
\put(13,0){
\path(0,0)(1,0)\ELB(0,0)\ELB(1,0)\ELB(2,1)
\path(3,2)(8,2)\ELB(3,2)\ELB(8,2)
\path(4,3)(7,3)\ELB(4,3)\ELB(6,3)\ELB(7,3)\ELB(5,4)}
\end{picture}
\end{center}
In particular, if $\rho=\{\varnothing\}$ then $(\sigma)_1$ is bounded by the
doublet $\{1,s+2\}$, and if $\sigma=\{\varnothing\}$ then the doublet
$\{r+1,r+2\}$ is appended to $\rho$. Notice that $\deg
M(\rho,\sigma)=\deg\rho\deg\sigma$.
\smallskip

{\em Operation $d_j$.} It consists of adding one element $n+1$ to
$\pi\in\Pi^*_n$. If $j=1$ then the element is added just as a singleton. For
$1<j\le k=\deg\pi$, the operation requires a detailed description. Let us
denote $\mu_2,\dots,\mu_k$ all multiplets in $\pi$, ordered by increase of
their minimal elements. Assume that all blocks with support containing $\mu_j$
are enumerated by a sequence $j_1<\dots<j_s=j$. Let us divide each of these
blocks into the left and right parts with respect to $m=\max\mu_j$:
\[
 \mu^-_{j_r}=\{i\in\mu_{j_r}: i<m\},\quad
 \mu^+_{j_r}=\{i\in\mu_{j_r}: i\ge m\}
\]
and form the new blocks
\[
 \tilde\mu_{j_1}=\mu^-_{j_1}\cup\{m,n+1\},\quad
 \tilde\mu_{j_r}=\mu^-_{j_r}\cup\mu^+_{j_{r-1}},~~r=2,\dots,s
\]
as shown on the following diagram. The rest blocks of the partition do not
change under this operation.
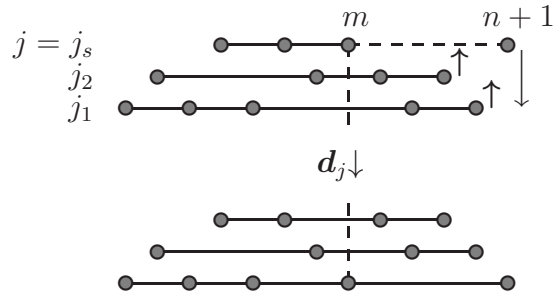
\begin{center}
\begin{picture}(15,9)(-1,-6)
{\thicklines\dashline[40]{0.4}(8,2)(13,2)%
  \dashline[40]{0.4}(8,-0.5)(8,2)}
\path(1,0)(12,0) \ELB(1,0)\ELB(3,0)\ELB(5,0)\ELB(10,0)\ELB(12,0)
\path(2,1)(11,1) \ELB(2,1)\ELB(7,1)\ELB(9,1)\ELB(11,1)
\path(4,2)(8,2)  \ELB(4,2)\ELB(6,2)\ELB(8,2)\ELB(13,2)
\put(7.8,2.6){$m$}\put(12.2,2.6){$n+1$}
\multiputlist(0,0)(0,1)[r]{$j_1$,$j_2$,$j=j_s$}%
\put(12.2,0.2){$\boldsymbol\uparrow$}
\put(11.2,1.2){$\boldsymbol\uparrow$}
\put(13.1,0.7){$\Big\downarrow$}
\put(7,-2){$\boldsymbol d_j\big\downarrow$}
\put(0,-5.5){{\thicklines\dashline[40]{0.4}(8,0)(8,2.5)}
\path(1,0)(13,0)\ELB(1,0)\ELB(3,0)\ELB(5,0)\ELB(8,0)\ELB(13,0)
\path(2,1)(12,1)\ELB(2,1)\ELB(7,1)\ELB(10,1)\ELB(12,1)
\path(4,2)(11,2)\ELB(4,2)\ELB(6,2)\ELB(9,2)\ELB(11,2)}
\end{picture}
\end{center}

\begin{theorem}
The operations $M$, $d_j$ generate any non-overlapping partition, in a unique
way.
\end{theorem}

\begin{proof}
The last operation bringing to a given partition is uniquely defined by
consideration of the block $\beta$ containing the maximal element of the
partition. If it is a singleton, then the last operation was $d_1$; if it is a
doublet, then it was $M$; if it is a multiplet, then the operation was $d_j$
where $j$ is the maximal number such that the support of multiplet $\mu_j$
contains the last to the end element of $\beta$. In each case, applying of
inverse operation brings to NOPs with lesser numbers of elements.
\end{proof}

Taking the theorem \ref{th:KdV1} into account, the established bijection allows
to associate a certain monomial with each NOP, although not in a quite
effective way, because we first have to build an exression $\varphi\in\Phi_n$
corresponding to $\pi\in\Pi^*_{n-1}$ and then to compute $\expand(\varphi)$:
\[
 \begin{array}{rccc}
   &\Phi_n & \leftrightarrow & \Pi^*_{n-1}\\
  \expand\!\!\!\! &\downarrow & \swarrow & \\
   &f_n &&
 \end{array}
\]
Nevertheless, it is easy to trace at the degree of monomial under this
correspondence; it is one more than the number of multiplets in the partition.
This gives us the following interpretation of the number triangle (\ref{fu}).

\begin{corollary}
The number of NOPs of $n$ elements containing $k$ multiplets is equal to the
number in the $n$-th row and $k$-th column of the number triangle (\ref{fu}),
starting their enumeration from 0.  This number is equal to the coefficient of
$u^{k+1}$ in the polynomial $F_{n+1}(u)=f_{n+1}(u,\dots,u)$ defined by the
recurrent relations
\[
 F_1=u,\quad
 F_{n+1}=u\partial_u(F_n)+\sum^{n-1}_{s=1}F_sF_{n-s},\quad n=1,2,\dots
\]
\end{corollary}

\section{Conclusion}

We have established a relation between several classes of set partitions and
generating functions for integrable hierarchies. Hopefully, this observation
may turn useful for both theories. Of course, we have too few examples at the
moment to make far-reaching conclusions. A conjecture is that each integrable
hierarchy has an underlying generating function which may be interpreted as
statistics for some kind of combinatorial objects (possibly unknown). As
further steps, it would be interesting to reveal the combinatorics associated
with the mKdV equation, KdV-like equations of 5-th order, nonlinear
Schr\"odinger equation and so on.

On the other hand, the objects studied in the combinatorics are so plentiful
and diverse that it seems doubtful that any one can be associated with an
integrable hierarchy. In all likelihood, this is a very special property, so it
would be interesting to understand what is the integrability intermediately in
combinatorial terms (rather than of the level of generating functions). In
particular, one can try to obtain a proof of commutativity of the flows of a
hierarchy based on their combinatorial interpretation.



\section*{References}
\begin{thebibliography}{99}

\bibitem{Alexandrov_Mironov_Morozov_Natanzon_2012}
 Alexandrov A, Mironov A, Morozov A and Natanzon S 2012
 Integrability of Hurwitz partition functions.
 {\it J. Phys. A: Math. Theor.}
 \href{http://dx.doi.org/10.1088/1751-8113/45/4/045209}{{\bf 45:4} 045209}

\bibitem{Benoumhani_1996} Benoumhani M 1996
 On Whitney numbers of Dowling lattices.
 {\it Discrete Mathematics}
 \href{http://dx.doi.org/10.1016/0012-365X(95)00095-E}{{\bf 159} 13--33}

\bibitem{Claesson_2001} Claesson A 2001
 Generalized pattern avoidance.
 {\it Europ. J. Combinatorics}
 \href{http://dx.doi.org/10.1006/eujc.2001.0515}{{\bf 22} 961--971}

\bibitem{Comtet_1974} Comtet L 1974
 {\it Advanced combinatorics. The art of finite and infinite expansions}
 (D. Reidel, Dordrecht)

\bibitem{Deift_2000} Deift P 2000
 Integrable systems and combinatorial theory.
 {\it Notices of the AMS}
 \href{http://www.ams.org/notices/200006/fea-deift.pdf}{{\bf 47:6} 631--640}

\bibitem{Dowling_1973} Dowling T A 1973
 A class of geometric lattices based on finite groups.
 {\it J. of Combin. Theory, Ser. B}
 \href{http://dx.doi.org/10.1016/S0095-8956(73)80007-3}{{\bf 14:1} 61--86}

\bibitem{Flajolet_Schott_1990} Flajolet P and Schott R 1990
 Non-overlapping partitions, continued fractions, Bessel functions and
 a divergent series.
 {\it Europ. J. Combinatorics}
 \href{http://dx.doi.org/10.1016/S0195-6698(13)80025-X}{{\bf 11:5} 421--432}

\bibitem{Gelfand_Dikii_1975} Gel'fand I M and Dikii L A 1975
 Asymptotic properties of the resolvent of Sturm--Liouville equations,
 and the algebra of Korteweg--de Vries equations.
 {\it Russian Math. Surveys}
 \href{http://dx.doi.org/10.1070%2FRM1975v030n05ABEH001522}{{\bf 30:5} 77--113}

\bibitem{Grosset_Veselov_2005} Grosset M-P and Veselov A P 2005
 Bernoulli numbers and solitons.
 {\it J. of Nonl. Math. Phys.}
 \href{http://dx.doi.org/10.2991/jnmp.2005.12.4.3}{{\bf 12:4} 469--474}

\bibitem{Ibragimov_Shabat_1980} Ibragimov N Kh and Shabat A B 1980
 Infinite Lie--B\"acklund algebras.
 {\it Funct. Anal. Appl.}
 \href{http://dx.doi.org/10.1007/BF01078315}{{\bf 14:4} 313--315}

\bibitem{Klazar_2003} Klazar M 2003
 Bell numbers, their relatives, and algebraic differential equations.
 {\it J. of Combin. Theory, Ser. A}
 \href{http://dx.doi.org/10.1016/S0097-3165(03)00014-1}{{\bf 102} 63--87}.

\bibitem{Lambert_Loris_Springael_Willox_1994}
 Lambert F, Loris I, Springael J and Willox R 1994
 On a direct bilinearization method: Kaup's higher-order water wave equation
 as a modified nonlocal Boussinesq equation.
 {\it J. Phys. A: Math. Gen.}
 \href{http://dx.doi.org/10.1088/0305-4470/27/15/028}{{\bf 27} 5325--5334}.

\bibitem{Lambert_Springael_2008} Lambert F and Springael J 2008
 Soliton equations and simple combinatorics.
 {\it Acta Appl. Math.}
 \href{http://dx.doi.org/10.1007/s10440-008-9209-3}{{\bf 102} 147--178}

\bibitem{Reiner_1997} Reiner V 1997
 Non-crossing partitions for classical reflection groups.
 {\it Discrete Mathematics}
 \href{http://dx.doi.org/10.1016/S0012-365X(96)00365-2}{{\bf 177:1--3} 195--222}

\bibitem{OEIS} Sloane N J A
 {\it The On-Line Encyclopedia of Integer Sequences,}
 published electronically at \url{http://oeis.org}, 2010
 (A000108 Catalan numbers; A000110 Bell numbers;
  A000295 Euler numbers; A005493 2-Bell numbers;
  A006789 Bessel numbers; A007405 Dowling numbers;
  A008277, A048993 Stirling numbers of the 2nd kind;
  A039755 $B$-analogs of Stirling numbers of the 2nd kind)

\bibitem{Schimming_1995} Schimming R 1995
 An explicit expression for the Korteweg--de Vries hierarchy.
 {\it Acta Appl. Math.}
 \href{http://dx.doi.org/10.1007/978-94-011-0017-5_28}{{\bf 39} 489--505}

\bibitem{Suter_2000} Suter R 2000
 Two analogues of a classical sequence.
 {\it J. Integer Seq.}
 \href{http://www.cs.uwaterloo.ca/journals/JIS/VOL3/SUTER/sut1.html}
 {{\bf 3} Article 00.1.8}.

\end{thebibliography}
\end{document}